\documentclass[12pt]{article}
\usepackage{mathpazo}      
\usepackage{charter}        
\usepackage{amssymb,amsmath,amsthm,amstext,graphics,amsfonts,hyperref,color}
\usepackage{cite}
\usepackage{graphicx}
\usepackage{fancyhdr}
\usepackage{enumerate}
\usepackage[justification=centering]{caption}
\usepackage[labelformat=simple]{subcaption}

\newcommand{\seq}{\subseteq}
\newcommand{\ZZ}{\mathbb{Z}}

\newcommand{\CC}{\mathbb{C}}

\newcommand{\DS}{\displaystyle}

\newcommand{\cv}{\mathbf{c}}
\newcommand{\dv}{\mathbf{d}}
\newcommand{\rv}{\mathbf{r}}
\newcommand{\sv}{\mathbf{s}}
\newcommand{\xv}{\mathbf{x}}
\newcommand{\yv}{\mathbf{y}}
\newtheorem{theorem}{Theorem}[section]
\newtheorem{lem}[theorem]{Lemma}
\newtheorem{cor}[theorem]{Corollary}

\theoremstyle{definition}
\newtheorem{ex}[theorem]{Example}

\newtheorem{remark}{Remark}

\begin{document}

\title{Linear codes over the ring
$\ZZ_4+u\ZZ_4+v\ZZ_4+w\ZZ_4+uv\ZZ_4+uw\ZZ_4+vw\ZZ_4+uvw\ZZ_4$
}
\author{Bustomi\footnote{Department of Mathematics,
Faculty of Sciences and Technology,
Universitas Airlangga,
Campus C Jl. Mulyorejo Surabaya, 60115
INDONESIA},
Aditya Purwa Santika\footnote{Combinatorial Mathematics Research Group,
Faculty of Mathematics and Natural Sciences,
Institut Teknologi Bandung,
Jl. Ganesha 10, Bandung, 40132,
INDONESIA},
and Djoko Suprijanto\footnote{Combinatorial Mathematics Research Group,
Faculty of Mathematics and Natural Sciences,
Institut Teknologi Bandung,
Jl. Ganesha 10, Bandung, 40132,
INDONESIA,\hfill \texttt{djoko@math.itb.ac.id} }
}

\maketitle

\begin{abstract}
We investigate linear codes over the ring $\ZZ_4+u\ZZ_4+v\ZZ_4+w\ZZ_4+uv\ZZ_4+uw\ZZ_4+vw\ZZ_4+uvw\ZZ_4,$
with conditions $u^2=u$, $v^2=v$, $w^2=w$, $uv=vu$, $uw=wu$ and $vw=wv.$
We first analyze the structure of the ring and then define linear codes over this ring.
Lee weight and Gray map for these codes are defined and MacWilliams relations for complete, symmetrized, and Lee weight
enumerators are obtained.  The Singleton bound as well as maximum distance separable codes are also considered.  Furthermore,
cyclic and quasi-cyclic codes are discussed, and some examples are also provided.

\vspace{0.25cm}

\textbf{Keywords:} Linear codes, MacWilliams relations, Maximum distance separable codes, Cyclic codes, Quasi-cyclic codes.
\\
\end{abstract}

\section{Introduction}
Codes over rings have become an active research area in  classical coding theory over the recent decades.
In particular, after the appearance of the work of Hammons, Kumar, Calderbank, Sloane, and Sol\'{e} \cite{Hammons},
a lot of research went towards studying (linear) codes over $\ZZ_4.$  Although "the results were generalized to
many different types of rings, the codes over $\ZZ_4$ remain a special topic of interest in the field of
algebraic coding theory because of their relation to lattices, designs, cryptography and their many
applications"\footnote{\cite{Yildiz2014}, pp. 25.} \cite{Yildiz2014}.

Recently, several new families of rings, namely the non-chain rings but are Frobenius,
have been studied in connection with coding theory.  These rings have rich mathematical theory, in particular
algebraic structures.  Yildiz and Karadeniz \cite{Yildiz2014} derived algebraic structures related to linear codes over
the ring $\ZZ_4+u\ZZ_4,$ with $u^2=0.$  They \cite{Yildiz2014} also have several good formally self-dual codes
over $\ZZ_4$  from the codes over $\ZZ_4+u\ZZ_4.$   Bandi and Bhaintwal (\cite{Bandi1}, \cite{Bandi2})  considered
codes over the ring $\ZZ_4+v\ZZ_4,$ with $v^2=v,$ and $\ZZ_4+w\ZZ_4,$ with $w^2=1,2w,$ respectively
and derived several algebraic structures including the MacWilliams relation with respect to Rosenbloom-Tsfasman metric
over the ring $\ZZ_4+v\ZZ_4$ and the properties as well as construction of self-dual codes over the ring $\ZZ_4+w\ZZ_4.$
Very recently, Li, Guo, Zhu, and Kai \cite{Li17} generalized the ring considered by Bandi and Bhaintwal \cite{Bandi1}
by adding two new terms $u\ZZ_4$ and $uv\ZZ_4,$ with the conditions $u^2=u,$ $v^2=v,$ and $uv=vu,$ and
derived some properties corresponding to the linear codes over the ring
$\ZZ_4+u\ZZ_4+v\ZZ_4+uv\ZZ_4.$

In this paper, we further generalized the ring considered by  Li, Guo, Zhu, and Kai \cite{Li17} to
the ring  $\ZZ_4+u\ZZ_4+v\ZZ_4+w\ZZ_4+uv\ZZ_4+uw\ZZ_4+vw\ZZ_4+uvw\ZZ_4$ with the conditions
$u^2=u$, $v^2=v$, $w^2=w$, $uv=vu$, $uw=wu$ and $vw=wv$.  We study linear codes over this ring
and derive some corresponding properties.  The paper is organized as follows. In Section 2, we study main properties
of the ring $\ZZ_4+u\ZZ_4+v\ZZ_4+w\ZZ_4+uv\ZZ_4+uw\ZZ_4+vw\ZZ_4+uvw\ZZ_4.$  We then define linear codes, Lee weight,
and also a Gray map for the linear codes over $\ZZ_4+u\ZZ_4+v\ZZ_4+w\ZZ_4+uv\ZZ_4+uw\ZZ_4+vw\ZZ_4+uvw\ZZ_4.$
Singleton bound as well as maximum distance separable codes are slightly considered.
In Section 3, several kind of weight enumerators are defined and related MacWilliams relations are derived. Finally,
in Section 4, cyclic and quasi-cyclic codes over $\ZZ_4+u\ZZ_4+v\ZZ_4+w\ZZ_4+uv\ZZ_4+uw\ZZ_4+vw\ZZ_4+uvw\ZZ_4$
are investigated and several examples are provided.

Throughout this paper, we follow standard definitions of undefined terms as used in many coding theory books
(e.g. \cite{Huffman}).

\section{Structures of linear codes over $R$}

Throughout this paper, $R$ denotes the ring
$\ZZ_4+u\ZZ_4+v\ZZ_4+w\ZZ_4+uv\ZZ_4+uw\ZZ_4+vw\ZZ_4+uvw\ZZ_4$ with
$u^2=u$, $v^2=v$, $w^2=w$, $uv=vu$, $uw=wu$ dan $vw=wv$.

\subsection{Structures of the ring $R$ and a Gray map}

Let $R$ denote the ring
$\ZZ_4+u\ZZ_4+v\ZZ_4+w\ZZ_4+uv\ZZ_4+uw\ZZ_4+vw\ZZ_4+uvw\ZZ_4$ with
$u^2=u$, $v^2=v$, $w^2=w$, $uv=vu$, $uw=wu$ dan $vw=wv$.
The ring $R$ can also be regarded as a quotient ring of a polynomial ring over $\ZZ_4,$
namely $\ZZ_4[u,v,w]/\langle u^2-u,v^2-v,w^2-w\rangle$. $R$ is commutative with identity.
The element $\eta \in R$ is called idempotent if $\eta^2=\eta.$  The elements $x$  and $y$ of $R$ is called orthogonal
if $xy=0.$

We will define a Gray map from the ring $R$ to $\ZZ_4^8.$  For that purpose we consider first a decomposition
of $R.$

Consider the idempotent elements of $R$ below
\[
\begin{aligned}
	\eta_1 &= 1-u-v-w+uv+uw+vw-uvw = (1-u)(1-v)(1-w),\\
	\eta_2 &= u-uv-uw+uvw = u(1-v)(1-w),\\
	\eta_3 &= v-uv-vw+uvw = (1-u)v(1-w),\\
	\eta_4 &= w-uw-vw+uvw = (1-u)(1-v)w,\\
	\eta_5 &= uv-uvw = uv(1-w),\\
	\eta_6 &= uw-uvw = u(1-v)w,\\
	\eta_7 &= vw-uvw = (1-u)vw,\\
	\eta_8 &= uvw.
\end{aligned}
\]
The above eight elements are also pairwise orthogonal, since $\eta_i\eta_j=0$ for $i\neq j,$  and also
satisfy $\sum\limits_{i=1}^{8} \eta_i=1$. Hence, by Chinese Remainder Theorem, we have
\[
R=R\eta_1 \oplus R\eta_2 \oplus R\eta_3 \oplus R\eta_4 \oplus R\eta_5 \oplus R\eta_6 \oplus R\eta_7 \oplus R\eta_8.
\]
Moreover, for any $r=a+bu+cv+dw+euv+fuw+gvw+huvw\in R$ with $a,b,c,d,e,f,g,h \in \mathbb{Z}_4,$ we have
\begin{align*}
	r &= r\sum\limits_{i=1}^{8} \eta_i\\
	  &= r\eta_1+r\eta_2+r\eta_3+r\eta_4+r\eta_5+r\eta_6+r\eta_7+r\eta_8\\
	  &= a\eta_1+(a+b)\eta_2+(a+c)\eta_3+(a+d)\eta_4+(a+b+c+e)\eta_5\\
	  &\quad+(a+b+d+f)\eta_6+(a+c+d+g)\eta_7+(a+b+c+d+e+f+g+h)\eta_8\\
	  &= r_1\eta_1+r_2\eta_2+r_3\eta_3+r_4\eta_4+r_5\eta_5+r_6\eta_6+r_7\eta_7+r_8\eta_8
\end{align*}
with
\begin{align*}
	r_1 &= a\\
	r_2 &= a+b\\
	r_3 &= a+c\\
	r_4 &= a+d\\
	r_5 &= a+b+c+e\\
	r_6 &= a+b+d+f\\
	r_7 &= a+c+d+g\\
	r_8 &= a+b+c+d+e+f+g+h,
\end{align*}
and hence $r_1,r_2,r_3,r_4,r_5,r_6,r_7,r_8 \in \mathbb{Z}_4$.
It is clear that the expression $r=r_1\eta_1+r_2\eta_2+r_3\eta_3+r_4\eta_4+r_5\eta_5+r_6\eta_6+r_7\eta_7+r_8\eta_8$ is unique.
Then we define the map $\phi$ from $R$ to $\mathbb{Z}_4^8$ by
\[
r \longmapsto (r_1,r_2,r_3,r_4,r_5,r_6,r_7,r_8).
\]
It is easy to see that $\phi$ defines an isomorphism.  We define a \emph{Gray map} as an extension of the map $\phi$ on $R^n$ as
\begin{align*}
	\Phi : R^n &\longrightarrow \mathbb{Z}_4^{8n}\\
	(c_0,c_1,\dots,c_{n-1}) &\longmapsto (r_{1,0},r_{1,1},\dots,r_{1,n-1},r_{2,0},r_{2,1},\dots,r_{2,n-1},\dots,
r_{8,0},r_{8,1},\dots,r_{8,n-1}),	
\end{align*}
where $c_i\in R$ and $r_{ji}\in \ZZ_4$ satisfying $c_i=\sum\limits_{j=1}^8 r_{ji}\eta_j.$

The Lee weight on $\ZZ_4,$ denoted by $w_L,$ is defined as
\[
w_L(x):=
\begin{cases}
0, & x=0,\\
2, & x=2,\\
1, & x=1 \text{ or }3.
\end{cases}
\]
Remembering the map $\phi: r \longmapsto (r_1,r_2,\ldots,r_8),$ we define the Lee weight on $R$ as
$w_L(r)=\sum_{i=1}^8 w_L(r_i).$  The Lee weight of a vector $\mathbf{c}=(c_0,c_1,\ldots,c_{n-1}) \in R^n$
is defined to be a rational sum of the Lee weight of its components, that is $w_L(\mathbf{c})=\sum_{i=0}^{n-1} w_L(c_i).$
Moreover, for any $\cv, \dv \in R^n,$ the Lee distance between $\cv$ and $\dv$ is defined as
$d_L(\cv,\dv)=w_L(\cv-\dv).$  Meanwhile, we have also another kind of weight and distance called Hamming weight
and Hamming distance, that defined as $w_H(\rv)=|\{j:~r_j \neq 0,~0 \leq j \leq n-1\}|$
and $d_H(\rv, \sv)=w_H(\rv-\sv),$ for all $\rv,\sv \in R^n,$ respectively.

\subsection{Linear Codes over $R$}

A nonempty subset $C \seq R^n$ is called \emph{linear code} over $R$ if $C$ is a submodule of $R.$ To define
a dual of the code $C,$ let us first define the Euclidean inner product on $R^n.$
Let $\xv=(x_0,x_1,\ldots,x_{n-1})$ and $\yv=(y_0,y_1,\ldots,y_{n-1})$ be two vectors in $R^n.$
The Euclidean inner product of $\xv$
and $\yv$ is defined as
\[
\langle \xv, \yv \rangle:=\sum_{j=1}^n x_j y_j,
\]
where the operations are performed in the ring $R.$

\emph{Dual} of the code $C \seq R^n$ is the code
\[
C^\perp=\{\xv \in R^n:~\langle \xv, \yv \rangle=0, \text{ for all }\yv \in C\}.
\]
Clearly, $C^\perp$ is also linear if $C$ is linear over $R.$  Since $R$ is a Frobenius ring, we also have
$|C|\cdot|C^\perp|=4^{8n}.$

Denote $\mathbf{r}=(r^{(0)},r^{(1)},\ldots,r^{(n-1)}) \in R^n,$ and
$\mathbf{r}^{(i)}=r_{i1} \eta_1+r_{i2}\eta_2+\cdots+r_{i8}\eta_8,$ for $0 \leq i \leq n-1.$  Then $\mathbf{r}$ can be
uniquely expressed as
\[
\mathbf{r}=\mathbf{r}_1 \eta_1+\mathbf{r}_2 \eta_2+ \cdots+ \mathbf{r}_8 \eta_8,
\]
where $\mathbf{r}_j=(r_{0j},r_{1j},\ldots,r_{n-1,j}) \in \ZZ_4^n,$ for $1 \leq j \leq 8.$  By using this expression then
the inner product of any two vectors $\mathbf{x}, \mathbf{y} \in R^n$ can be written as
\[
\mathbf{x} \cdot \mathbf{y}=(\mathbf{x}_1 \cdot \mathbf{y}_1)\eta_1+(\mathbf{x}_2 \cdot \mathbf{y}_2)\eta_2+\cdots+
(\mathbf{x}_8 \cdot \mathbf{y}_8)\eta_8,
\]
where $\mathbf{x}=\mathbf{x}_1 \eta_1+\mathbf{x}_2 \eta_2+\cdots+\mathbf{x}_8 \eta_8,$
$\mathbf{x}_j=(x_{0j},x_{1j},\ldots,x_{n-1,j} ) \in \ZZ_4^n,$ and
$\mathbf{y}=\mathbf{y}_1 \eta_1+\mathbf{y}_2 \eta_2+\cdots+\mathbf{y}_8 \eta_8,$
$\mathbf{y}_j=(y_{0j},y_{1j},\ldots,y_{n-1,j} ) \in \ZZ_4^n,$ and $\mathbf{x}_j \cdot \mathbf{y}_j=\sum_{k=0}^{n-1}
x_{kj} y_{kj},$ for $1 \leq j \leq n.$

Now, define the codes $C_i,$ $1 \leq i \leq 8,$ as follows:
\[
\begin{aligned}
C_1 &=\{\mathbf{a} \in \ZZ_4^n:~\mathbf{a} \eta_1+\mathbf{b} \eta_2+\cdots+\mathbf{h} \eta_8,
\text{ for some } \mathbf{b}, \mathbf{c},\ldots,\mathbf{h} \in \ZZ_4^n \},\\
C_2 &=\{\mathbf{b} \in \ZZ_4^n:~\mathbf{a} \eta_1+\mathbf{b} \eta_2+\cdots+\mathbf{h} \eta_8,
\text{ for some } \mathbf{a}, \mathbf{c},\ldots,\mathbf{h} \in \ZZ_4^n \},\\
& \qquad \vdots\\
C_8 &=\{\mathbf{h} \in \ZZ_4^n:~\mathbf{a} \eta_1+\mathbf{b} \eta_2+\cdots+\mathbf{h} \eta_8,
\text{ for some } \mathbf{a}, \mathbf{b},\ldots,\mathbf{g} \in \ZZ_4^n \}.
\end{aligned}
\]
It is easy to see that $C_i,$ $1 \leq i \leq 8,$ is a linear code of length $n$ over $\ZZ_n.$  Moreover, $C$ can be
uniquely decomposed into $C=C_1 \eta_1 \oplus C_2 \eta_2 \oplus \cdots \oplus C_8 \eta_8,$
and hence we have $|C|=\prod_{i=1}^8 |C_i|.$
Furthermore, we have the following property.

\begin{theorem}\label{T-decom}
Let $C$ be a linear code of length $n$ over $R.$  Then we have the following unique decomposition:
\begin{enumerate}
\item $C=C_1 \eta_1 \oplus C_2 \eta_2 \oplus \cdots \oplus C_8 \eta_8,$ a linear code of length $n$ over $\ZZ_4.$

\item $C^\perp=C_1^\perp \eta_1 \oplus C_2^\perp \eta_2 \oplus \cdots \oplus C_8^\perp \eta_8,$ for $1 \leq i \leq 8.$

\end{enumerate}
\end{theorem}

\begin{proof}
Similar to \cite{Li17}.
\end{proof}

It is well-known (see for instance \cite{Hammons}) that the code $C_j,$ $1 \leq i \leq 8,$
is permutation-equivalent to a code generated by
\[
G_j=
\begin{pmatrix}
I_{k_{j1}} & A_j & B_j\\
0 & 2I_{k_{j2}}  & 2C_j
\end{pmatrix},
\]
where $A_j$ and $C_j$ are $\ZZ_2$-matrices and $B$ is a $\ZZ_4$-matrix,
and hence $C$ is permutation-equivalent to a linear code generated by
\[
\begin{pmatrix}
\eta_1 G_1\\
\eta_2 G_2\\
\vdots \\
\eta_8 G_8
\end{pmatrix}.
\]

Moreover, by the result in \cite{Hammons}, the dual $C_j^\perp$ has generator
\[
G_j^\prime=
\begin{pmatrix}
-B_i^T-C_i^TA_i^T & C_i^T & I_{n-k_{i1}-k_{i2}}\\
2A_i^T & 2I_{k_{i2}}  & 0
\end{pmatrix},  \text{ for }1 \leq i \leq 8,
\]
and hence $C^\perp$ is permutation-equivalent to a linear code generated by
\[
\begin{pmatrix}
\eta_1 G_1^\prime\\
\eta_2 G_2^\prime\\
\vdots \\
\eta_8 G_8^\prime
\end{pmatrix},
\]
a parity check matrix of the code $C.$

\begin{ex}
Let $C_j$ be a linear code over $\ZZ_4$ with generator matrix $G_j,$ for $1 \leq j \leq 8,$ as follows
\[
G_j=
\begin{pmatrix}
1 & 1 & 1 & 3\\
0 & 2 & 0 & 2\\
0 & 0 & 2 & 2
\end{pmatrix}.
\]
The code $C_j$ contains $2^{2k_{j1}+k_{j2}}=2^{2+2}=16.$ Since $4^4=|C| \cdot |C^\perp|=4^2 \cdot |C^\perp|,$
then $|C^\perp|=4^2=|C|.$ The linear code $\DS C= \bigoplus_{j=1}^8 C_j \eta_j$ over $R$
is of cardinality $|C|=16^8=4^{2 \cdot 8},$  while the generator matrix of $\Phi(C)$ is

\[
\left(
\begin{smallmatrix}
1 & 0 & 0 & 0 & 0 & 0 & 0 & 0 & 1 & 0 & 0 & 0 & 0 & 0 & 0 & 0 & 1 & 0 & 0 & 0 & 0 & 0 & 0 & 0 & 3 & 0 & 0 & 0 & 0 & 0 & 0 & 0\\
0 & 0 & 0 & 0 & 0 & 0 & 0 & 0 & 2 & 0 & 0 & 0 & 0 & 0 & 0 & 0 & 0 & 0 & 0 & 0 & 0 & 0 & 0 & 0 & 2 & 0 & 0 & 0 & 0 & 0 & 0 & 0\\
0 & 0 & 0 & 0 & 0 & 0 & 0 & 0 & 0 & 0 & 0 & 0 & 0 & 0 & 0 & 0 & 2 & 0 & 0 & 0 & 0 & 0 & 0 & 0 & 2 & 0 & 0 & 0 & 0 & 0 & 0 & 0\\
0 & 1 & 0 & 0 & 0 & 0 & 0 & 0 & 0 & 1 & 0 & 0 & 0 & 0 & 0 & 0 & 0 & 1 & 0 & 0 & 0 & 0 & 0 & 0 & 0 & 3 & 0 & 0 & 0 & 0 & 0 & 0\\
0 & 0 & 0 & 0 & 0 & 0 & 0 & 0 & 0 & 2 & 0 & 0 & 0 & 0 & 0 & 0 & 0 & 0 & 0 & 0 & 0 & 0 & 0 & 0 & 0 & 2 & 0 & 0 & 0 & 0 & 0 & 0\\
0 & 0 & 0 & 0 & 0 & 0 & 0 & 0 & 0 & 0 & 0 & 0 & 0 & 0 & 0 & 0 & 0 & 2 & 0 & 0 & 0 & 0 & 0 & 0 & 0 & 2 & 0 & 0 & 0 & 0 & 0 & 0\\
0 & 0 & 1 & 0 & 0 & 0 & 0 & 0 & 0 & 0 & 1 & 0 & 0 & 0 & 0 & 0 & 0 & 0 & 1 & 0 & 0 & 0 & 0 & 0 & 0 & 0 & 3 & 0 & 0 & 0 & 0 & 0\\
0 & 0 & 0 & 0 & 0 & 0 & 0 & 0 & 0 & 0 & 2 & 0 & 0 & 0 & 0 & 0 & 0 & 0 & 0 & 0 & 0 & 0 & 0 & 0 & 0 & 0 & 2 & 0 & 0 & 0 & 0 & 0\\
0 & 0 & 0 & 0 & 0 & 0 & 0 & 0 & 0 & 0 & 0 & 0 & 0 & 0 & 0 & 0 & 0 & 0 & 2 & 0 & 0 & 0 & 0 & 0 & 0 & 0 & 2 & 0 & 0 & 0 & 0 & 0\\
0 & 0 & 0 & 1 & 0 & 0 & 0 & 0 & 0 & 0 & 0 & 1 & 0 & 0 & 0 & 0 & 0 & 0 & 0 & 1 & 0 & 0 & 0 & 0 & 0 & 0 & 0 & 3 & 0 & 0 & 0 & 0\\
0 & 0 & 0 & 0 & 0 & 0 & 0 & 0 & 0 & 0 & 0 & 2 & 0 & 0 & 0 & 0 & 0 & 0 & 0 & 0 & 0 & 0 & 0 & 0 & 0 & 0 & 0 & 2 & 0 & 0 & 0 & 0\\
0 & 0 & 0 & 0 & 0 & 0 & 0 & 0 & 0 & 0 & 0 & 0 & 0 & 0 & 0 & 0 & 0 & 0 & 0 & 2 & 0 & 0 & 0 & 0 & 0 & 0 & 0 & 2 & 0 & 0 & 0 & 0\\
0 & 0 & 0 & 0 & 1 & 0 & 0 & 0 & 0 & 0 & 0 & 0 & 1 & 0 & 0 & 0 & 0 & 0 & 0 & 0 & 1 & 0 & 0 & 0 & 0 & 0 & 0 & 0 & 3 & 0 & 0 & 0\\
0 & 0 & 0 & 0 & 0 & 0 & 0 & 0 & 0 & 0 & 0 & 0 & 2 & 0 & 0 & 0 & 0 & 0 & 0 & 0 & 0 & 0 & 0 & 0 & 0 & 0 & 0 & 0 & 2 & 0 & 0 & 0\\
0 & 0 & 0 & 0 & 0 & 0 & 0 & 0 & 0 & 0 & 0 & 0 & 0 & 0 & 0 & 0 & 0 & 0 & 0 & 0 & 2 & 0 & 0 & 0 & 0 & 0 & 0 & 0 & 2 & 0 & 0 & 0\\
0 & 0 & 0 & 0 & 0 & 1 & 0 & 0 & 0 & 0 & 0 & 0 & 0 & 1 & 0 & 0 & 0 & 0 & 0 & 0 & 0 & 1 & 0 & 0 & 0 & 0 & 0 & 0 & 0 & 3 & 0 & 0\\
0 & 0 & 0 & 0 & 0 & 0 & 0 & 0 & 0 & 0 & 0 & 0 & 0 & 2 & 0 & 0 & 0 & 0 & 0 & 0 & 0 & 0 & 0 & 0 & 0 & 0 & 0 & 0 & 0 & 2 & 0 & 0\\
0 & 0 & 0 & 0 & 0 & 0 & 0 & 0 & 0 & 0 & 0 & 0 & 0 & 0 & 0 & 0 & 0 & 0 & 0 & 0 & 0 & 2 & 0 & 0 & 0 & 0 & 0 & 0 & 0 & 2 & 0 & 0\\
0 & 0 & 0 & 0 & 0 & 0 & 1 & 0 & 0 & 0 & 0 & 0 & 0 & 0 & 1 & 0 & 0 & 0 & 0 & 0 & 0 & 0 & 1 & 0 & 0 & 0 & 0 & 0 & 0 & 0 & 3 & 0\\
0 & 0 & 0 & 0 & 0 & 0 & 0 & 0 & 0 & 0 & 0 & 0 & 0 & 0 & 2 & 0 & 0 & 0 & 0 & 0 & 0 & 0 & 0 & 0 & 0 & 0 & 0 & 0 & 0 & 0 & 2 & 0\\
0 & 0 & 0 & 0 & 0 & 0 & 0 & 0 & 0 & 0 & 0 & 0 & 0 & 0 & 0 & 0 & 0 & 0 & 0 & 0 & 0 & 0 & 2 & 0 & 0 & 0 & 0 & 0 & 0 & 0 & 2 & 0\\
0 & 0 & 0 & 0 & 0 & 0 & 0 & 1 & 0 & 0 & 0 & 0 & 0 & 0 & 0 & 1 & 0 & 0 & 0 & 0 & 0 & 0 & 0 & 1 & 0 & 0 & 0 & 0 & 0 & 0 & 0 & 3\\
0 & 0 & 0 & 0 & 0 & 0 & 0 & 0 & 0 & 0 & 0 & 0 & 0 & 0 & 0 & 2 & 0 & 0 & 0 & 0 & 0 & 0 & 0 & 0 & 0 & 0 & 0 & 0 & 0 & 0 & 0 & 2\\
0 & 0 & 0 & 0 & 0 & 0 & 0 & 0 & 0 & 0 & 0 & 0 & 0 & 0 & 0 & 0 & 0 & 0 & 0 & 0 & 0 & 0 & 0 & 2 & 0 & 0 & 0 & 0 & 0 & 0 & 0 & 2
\end{smallmatrix}
\right).
\]
$\diamondsuit$
\end{ex}

\subsection{Singleton bound and MDS codes}

Singleton bound is among the famous bound in Coding Theory.  It is proven in 1964 by
Singleton \cite{Singleton} that if $C \seq R^n$ is a code over $R,$
then we have
\[
d_H(C)\leq n - \log_{|R|}|C|+1,
\]
where $d_H(C)=\text{min}\{w_H(\xv-\yv):~ \xv,\yv \in C, \xv \neq \yv\}.$  The code $C$ is called
\emph{maximum distance separable} (MDS) if it attains the Singleton bound mentioned above.

It has been proven by Guenda and Gulliver \cite[Proposition 2.2]{Guenda} that the only MDS codes
over $\ZZ_4$ is the trivial one.  Moreover, it is also known that
$C^\perp$ is an MDS code if $C$ is an MDS code (see \cite[Theorem 1]{Shiro2}). Hence, we have the following.

\begin{lem}
Let $C$ be a linear code of length $n$ over $\ZZ_4.$  Then $C$ is an MDS codes if and only if
$C$ is either $\ZZ_4^n$ of parameters $[n,4^n,1],$ $\langle \mathbb{1} \rangle$ of parameters $[n,4,n],$
or $\langle \mathbb{1} \rangle^\perp$ of parameters $[n,4^{n-1},2],$ where $\mathbb{1}$ denotes the all-one vector.
\end{lem}

Let us look at the MDS codes over $R.$  By considering a linear code $C$ of length $n$ over $R$ as
$C=C_1 \eta_1 \oplus C_2 \eta_2 \oplus \cdots \oplus C_8 \eta_8,$  where $C_i,$ $1 \leq i \leq 8,$  is a linear code of
length $n$ over $\ZZ_4,$ the Singleton bound can be written as
\[
d_H(C) \leq n-\frac{1}{8} \sum_{i=1}^8 \log_{4}|C_i| + 1,
\]
where $d_H(C)=\text{min}\{d_H(C_i):~1 \leq i \leq 8\},$ and $d_H(C_i)$ is a Hamming distance of $C_i,$
for $1 \leq i \leq 8.$

Then we have the following theorem.

\begin{theorem}
Let $C$ be an MDS codes of length $n$ over $R.$
\begin{itemize}
\item[(1)] If $d_H=1,$ then $C_i$ is an MDS code of parameters $[n,4^n,1],$ for $1 \leq i \leq 8.$

\item[(2)] If $d_H=2,$ then $C_i$ is an MDS code of parameters $[n,4^{n-1},2],$ for $1 \leq i \leq 8.$

\item[(3)] If $d_H=n,$ then $C_i$ is an MDS code of parameters $[n,4,n],$ for $1 \leq i \leq 8.$

\end{itemize}

\end{theorem}

\begin{proof}
Similar to the proof of Theorem 5 in \cite{Li17}.
\end{proof}

\begin{theorem}
$C$ is an MDS code of length $n$ over $R$ if and only if for $1 \leq i \leq 8,$
$C_i$ is an MDS code over $\ZZ_4$ with the same parameters.
\end{theorem}

\begin{proof}
Similar to the proof of Theorem 6 in \cite{Li17}.
\end{proof}

\section{Weight enumerators and MacWilliams relations} \label{Macwill}

In this section we consider several weight enumerators for a linear codes $C$ as well as
the related MacWilliams relations.

\subsection{The complete weight enumerator and MacWilliams relation}

We knew that the number of elements of $R$ is $65536.$

The \emph{complete weight enumerator (CWE)} of a linear code $C \seq R^n$ is defined as
\[
\begin{aligned}
CWE_C(X_0,X_1,\ldots,X_{65535})&=\sum_{\cv \in C} X_0^{n_{a_0}(\cv)}X_1^{n_{a_1}(\cv)}\cdots X_{65535}^{n_{a_{65535}}(\cv)}\\
&=\sum\limits_{\cv \in C} \prod\limits_{j=0}^{65535}   X_j^{n_{a_j}(\cv)},
\end{aligned}
\]
where $n_{a_i}(\cv)$ denotes the number of appearances of $a_i \in R$ in the vector $\cv.$

\begin{remark}
Note that $CWE_C(X_0,X_1,\ldots,X_{65535})$ is a homogeneous polynomial in $65536$ variables with total
degree of each monomial being $n,$ the length of the code $C.$  Since the code $C$ is linear, then $C$
always contains the vector $\mathbf{0}.$  It implies that the term $X_0^n$ always appears in
$CWE_C(X_0,X_1,\ldots,X_{65535}).$  From the complete weight enumerator we may obtain
a lot of information related to the code, such as the size of the code:

\[
CWE_C(1,1,\ldots,1)=\sum_{\cv \in C} 1=|C|.
\]
\qed
\end{remark}

Since the ring $R$ is a Frobenius ring, then the MacWilliams relation for the complete weight enumerator
holds (see \cite{Wood}).  To find the exact relation we define the following character on $R.$

Let $I$ be a non-zero ideal in $R$. Define $\chi : I \rightarrow \mathbb{C}^*$ by
$$\chi(a+bu+cv+dw+euv+fuw+gvw+huvw)=i^h$$
with $\CC^*$ is a unit group in complex number. We know that $\chi$ is a non-trivial character on $R.$\\

Defining the Hadamard transform by
$\hat{f}(\textbf{c})=\sum\limits_{\textbf{d}\in R^n} \chi(\textbf{c}\cdot \textbf{d})f(\textbf{d}),$
we obtain the following equation
\begin{equation}\label{P-Hadamard}
	\sum_{\cv \in C}\hat{f}(\cv)=|C|\sum\limits_{\mathbf{d}\in C^\perp} f(\mathbf{d}).
\end{equation}

We have the MacWilliams relation with respect to the complete weight enumerator as follows.

\begin{theorem}
Let $C$ be a linear code of length $n$ over $R.$  Then
\[
CWE_C(X_0,X_1,\ldots,X_{65535})=\dfrac{1}{|C|}CWE_C(M(X_0,X_1,\ldots,X_{65535})^T)
\]
with $M$ is a matrix of size $65536 \times 65536$ defined by $M_{ij}=\chi(a_ia_j)$.
\end{theorem}

\begin{proof}
Let $f(x)=\prod\limits_{i=0}^{65535}X_i^{n_{a_i}(x)}$. The result follows from Theorem 8.1 in \cite{Wood}.
\end{proof}

\subsection{The Symmetrized Lee weight, Hamming weight and Lee weight enumerator}

In the ring $\ZZ_4,$ we know that $w_L(1)=1=w_L(3)$ and the symmetrized Lee weight enumerator for codes over $\ZZ_4$
is defined as
\[
SLWE_C(X_0,X_1,X_2)=CWE_C(X_0,X_1,X_2,X_1).
\]
Adopting the same idea, we will define the symmetrized Lee weight enumerator of codes over $R.$
For that purpose, we first decompose
$R$ into $D_i=\{x\in R:~w_L(x)=i\},$ for $0 \leq i \leq 16.$  Then we have
\[
\begin{aligned}
	|D_{0}|&=|D_{16}|=1,\\
	|D_{1}|&=|D_{15}|=2 \binom{8}{1}=16,\\
	|D_{2}|&=|D_{14}|=2^2 \binom{8}{2}+\binom{8}{1}=120,\\
	|D_{3}|&=|D_{13}|=2^3 \binom{8}{3}+2 \binom{8}{1} \binom{7}{1}=560,\\
	|D_{4}|&=|D_{12}|=2^4 \binom{8}{4}+2^2 \binom{8}{2} \binom{6}{1} +\binom{8}{2}=1820,\\
	|D_{5}|&=|D_{11}|=2^5 \binom{8}{5}+2^3 \binom{8}{3} \binom{5}{1} +2 \binom{8}{1} \binom{7}{2}=4368,\\
	|D_{6}|&=|D_{10}|=2^6 \binom{8}{6}+2^4 \binom{8}{4} \binom{4}{1} +2^2 \binom{8}{2} \binom{6}{2}+\binom{8}{3}=8008,\\
	|D_{7}|&=|D_{9}|=2^7 \binom{8}{7}+2^5 \binom{8}{5} \binom{3}{1} +2^3 \binom{8}{3} \binom{5}{2}+2\binom{8}{1} \binom{7}{3}=11440,\\
	|D_{8}|&=2^8 \binom{8}{8}+2^6 \binom{8}{6} \binom{2}{1} +2^4\binom{8}{4} \binom{4}{2}+2^2 \binom{8}{2} \binom{6}{3}+\binom{8}{4}=12870.
\end{aligned}
\]

By looking at the elements that have the same Lee weights, we can define the symmetrized Lee weight enumerator.
\emph{Symmetrized Lee weight enumerator (SLWE)} of a linear code $C$ over $R$  is defined as
\[
\begin{aligned}
SLWE_C &(X_0,X_1,\ldots,X_{16})
=CLWE_C(X_0,\underbrace{X_1,\ldots,X_1}_{16},\underbrace{X_2,\ldots,X_2}_{120},\underbrace{X_3,\ldots,X_3}_{560},\\
&\underbrace{X_4,\ldots,X_4}_{1820},\underbrace{X_5,\ldots,X_5}_{4368},\underbrace{X_6,\ldots,X_6}_{8008},
\underbrace{X_7,\ldots,X_7}_{11440},\underbrace{X_8,\ldots,X_8}_{12870},\underbrace{X_9,\ldots,X_9}_{11440},\\
&\underbrace{X_{10},\ldots,X_{10}}_{8008},\underbrace{X_{11},\ldots,X_{11}}_{4368},
\underbrace{X_{12},\ldots,X_{12}}_{1820},\underbrace{X_{13},\ldots,X_{13}}_{560},
\underbrace{X_{14},\ldots,X_{14}}_{120},\\
&\underbrace{X_{15},\ldots,X_{15}}_{16},X_{16})
\end{aligned}
\]
where $X_0,X_1,X_2,\ldots,X_{16}$ denote the element of weight $0,1,2,3,\ldots,16,$ respectively. Then we have
\[
\begin{aligned}
	SLWE_C &(X_0,X_1,X_2,\ldots,X_{16})=
\sum\limits_{\cv \in C}X_0^{n_{0}(\cv)}X_1^{n_{1}(\cv)}X_2^{n_{2}(\cv)} \ldots X_{16}^{n_{16}(\cv)},
\end{aligned}
\]
where
\[
\begin{aligned}
	n_{0}&=n_{a_1}(\cv), \quad &	
    n_{1}&=\sum\limits_{i=2}^{17}n_{a_i}(\cv), \quad&
    n_{2}&=\sum\limits_{i=18}^{137}n_{a_i}(\cv),\\	
    n_{3}&=\sum\limits_{i=138}^{697}n_{a_i}(\cv), \quad&	
    n_{4}&=\sum\limits_{i=698}^{2517}n_{a_i}(\cv), \quad&	
    n_{5}&=\sum\limits_{i=2518}^{6885}n_{a_i}(\cv),\\
	n_{6}&=\sum\limits_{i=6886}^{14893}n_{a_i}(\cv),\quad&	
    n_{7}&=\sum\limits_{i=15434}^{26333}n_{a_i}(\cv),\quad&
	n_{8}&=\sum\limits_{i=26334}^{39203}n_{a_i}(\cv),\\
	n_{9}&=\sum\limits_{i=39204}^{50643}n_{a_i}(\cv),\quad&
	n_{10}&=\sum\limits_{i=50644}^{58651}n_{a_i}(\cv),\quad&
	n_{11}&=\sum\limits_{i=58652}^{64839}n_{a_i}(\cv),\\
	n_{12}&=\sum\limits_{i=64840}^{64799}n_{a_i}(\cv),\quad&
	n_{13}&=\sum\limits_{i=64800}^{65399}n_{a_i}(\cv),\quad&
	n_{14}&=\sum\limits_{i=65400}^{65519}n_{a_i}(\cv),\\
	n_{15}&=\sum\limits_{i=65520}^{65535}n_{a_i}(\cv),\quad&
	n_{16}&=n_{a_{65536}}(\cv).
\end{aligned}
\]

The MacWilliams relation with respect to the symmetrized Lee weight enumerator is as follows.

\begin{theorem}\label{slwe}
Let $C$ be a linear code of length $n$ over $R$. Then
\[	
SLWE_{C^\perp} (X_0,X_1,\ldots,X_{16})=\dfrac{1}{|C|}SLWE_C (B_0,B_1,\ldots,B_{16}),
\]	
where
\[	
\begin{aligned}
		B_0 = ~&X_0+16X_1+120X_2+560X_3+1280X_4+4368X_5+8008X_6+11440X_7+12870X_8\\
		&+11440X_9+8008X_{10}+4368X_{11}+1280X_{12}+560X_{13}+120X_{14}+16X_{15}+X_{16},\\
		B_1 = ~&X_0+14X_1+90X_2+350X_3+910X_4+1638X_5+2002X_6+1430X_7-1430X_9\\
        &-2002X_{10}-1638X_{11}-910X_{12}-350X_{13}-90X_{14}-14X_{15}-X_{16},\\
		B_2 = ~&X_0+12X_1+64X_2+196X_3+364X_4+364X_5-572X_7-858X_8-572X_9+364X_{11}\\
		&+364X_{12}+196X_{13}+64X_{14}+12X_{15}+X_{16},\\
		B_3 = ~&X_0+10X_1+42X_2+90X_3+78X_4-78X_5-286X_6-286X_7+286X_9+286X_{10}\\
        &+78X_{11}-78X_{12}-90X_{13}-42X_{14}-10X_{15}-X_{16},\\
		B_4 = ~&X_0+8X_1+24X_2+24X_3-36X_4-120X_5-88X_6+88X_7+198X_8+88X_9\\
        &-88X_{10}-120X_{11}-36X_{12}+24X_{13}+24X_{14}+8X_{15}+X_{16},\\
		B_5 = ~&X_0+6X_1+10X_2-10X_3-50X_4-34X_5+66X_6+110X_7-110X_9-66X_{10}\\
        &+34X_{11}+50X_{12}+10X_{13}-10X_{14}-6X_{15}-X_{16},\\
		B_6 = ~&X_0+4X_1-20X_3-20X_4+36X_5+64X_6-20X_7-90X_8-20X_9+64X_{10}\\
        &+36X_{11}-20X_{12}-20X_{13}+4X_{15}+X_{16},\\
\end{aligned}
\]
\[
\begin{aligned}
		B_7 = ~&X_0+2X_1-6X_2-14X_3+14X_4+42X_5-14X_6-70X_7+70X_9+14X_{10}\\
        &-42X_{11}-14X_{12}+14X_{13}+6X_{14}-2X_{15}-X_{16},\\
		B_8 = ~&X_0-8X_2+28X_4-56X_6+70X_8-56X_{10}+28X_{12}-8X_{14}+X_{16},\\
		B_9 = ~&X_0-2X_1-6X_2+14X_3+14X_4-42X_5-14X_6+70X_7-70X_9+14X_{10}\\
        &+42X_{11}-14X_{12}-14X_{13}+6X_{14}+2X_{15}-X_{16},\\
		B_{10} = ~&X_0-4X_1+20X_3-20X_4-36X_5+64X_6+20X_7-90X_8+20X_9+64X_{10}\\
        &-36X_{11}-20X_{12}+20X_{13}-4X_{15}+X_{16},\\
		B_{11} = ~&X_0-6X_1+10X_2+10X_3-50X_4+34X_5+66X_6-110X_7+110X_9-66X_{10}\\
        &-34X_{11}+50X_{12}-10X_{13}-10X_{14}+6X_{15}-X_{16},\\
		B_{12} = ~&X_0-8X_1+24X_2-24X_3-36X_4+120X_5-88X_6-88X_7+198X_8-88X_9\\
        &-88X_{10}+120X_{11}-36X_{12}-24X_{13}+24X_{14}-8X_{15}+X_{16},\\
		B_{13} = ~&X_0-10X_1+42X_2-90X_3+78X_4+78X_5-286X_6+286X_7-286X_9+286X_{10}\\
        &-78X_{11}-78X_{12}+90X_{13}-42X_{14}+10X_{15}-X_{16},\\
		B_{14} = ~&X_0-12X_1+64X_2-196X_3+364X_4-364X_5+572X_7-858X_8+572X_9-364X_{11}\\
        &+364X_{12}-196X_{13}+64X_{14}-12X_{15}+X_{16},\\
		B_{15} = ~&X_0-14X_1+90X_2-350X_3+910X_4-1638X_5+2002X_6-1430X_7+1430X_9\\
        &-2002X_{10}+1638X_{11}-910X_{12}+350X_{13}-90X_{14}+14X_{15}-X_{16},\\
		B_{16} = ~&X_0-16X_1+120X_2-560X_3+1280X_4-4368X_5+8008X_6-11440X_7+12870X_8\\
		&-11440X_9+8008X_{10}-4368X_{11}+1280X_{12}-560X_{13}+120X_{14}-16X_{15}+X_{16},\\
	\end{aligned}
\]
\end{theorem}

\begin{proof}
For $i,j=0,1,2,\ldots,16,$ we determine $\sum\limits_{s\in D_j} \chi(rs)$ for $r\in D_i$. By definition, we have
\[
\begin{aligned}
SLWE_{C^\perp} &(X_0,X_1,\ldots,X_{16})
=CWE_{C^\perp}(X_0,\underbrace{X_1,\ldots,X_1}_{16},\underbrace{X_2,\ldots,X_2}_{120},\underbrace{X_3,\ldots,X_3}_{560},
\underbrace{X_4,\ldots,X_4}_{1820},\\
&\underbrace{X_5,\ldots,X_5}_{4368},\underbrace{X_6,\ldots,X_6}_{8008},
\underbrace{X_7,\ldots,X_7}_{11440},
\underbrace{X_8,\ldots,X_8}_{12870},\underbrace{X_9,\ldots,X_9}_{11440},
\underbrace{X_{10},\ldots,X_{10}}_{8008},\\
&\underbrace{X_{11},\ldots,X_{11}}_{4368},
\underbrace{X_{12},\ldots,X_{12}}_{1820},\underbrace{X_{13},\ldots,X_{13}}_{560},
\underbrace{X_{14},\ldots,X_{14}}_{120},\underbrace{X_{15},\ldots,X_{15}}_{16},X_{16})\\
&=\dfrac{1}{|C|} CWE_{C}\left(\sum\limits_{j=0}^{16}\sum\limits_{s\in D_j}\chi(a_1s)X_j,
\sum\limits_{j=0}^{16}\sum\limits_{s\in D_j}\chi(a_2s)X_j,\dots,\sum\limits_{j=0}^{16}
\sum\limits_{s\in D_j}\chi(a_{65536}s)X_j\right).
\end{aligned}
\]

Since for $a_j,a_k\in D_j$ we have
\[
\sum\limits_{j=0}^{16}\sum\limits_{s\in D_j}\chi(a_js)X_j=\sum\limits_{j=0}^{16}\sum\limits_{s\in D_j}\chi(a_ks)X_j,
\]
then
\[
\begin{aligned}
		&SLWE_{C^\perp} (X_0,X_1,\ldots,X_{16})\\
		&=\dfrac{1}{|C|} SLWE_{C}\left(\sum\limits_{a_i\in D_0}\sum\limits_{j=0}^{16}
        \sum\limits_{s\in D_j}\chi(a_is)X_j,\dots,\sum\limits_{a_i\in D_{16}}\sum\limits_{j=0}^{16}
        \sum\limits_{s\in D_j}\chi(a_is)X_j\right).
\end{aligned}
\]
By direct calculation, we obtain
$$\sum\limits_{a_i\in D_k}\sum\limits_{j=0}^{16}\sum\limits_{s\in D_j}\chi(a_is)X_j=B_k$$ for $k=1,2,\ldots,16.$ Hence, we
have
$$SLWE_{C^\perp} (X_0,X_1,\ldots,X_{16})
=\dfrac{1}{|C|}SLWE_C (B_0,B_1,\ldots,B_{16}).$$
\end{proof}

Another weight enumerator of a linear code $C,$ called \emph{Hamming weight enumerator,}  is defined as
\[
Ham_C(X,Y)\sum\limits_{\cv \in C} X^{n-w_H(\cv)} Y^{w_H(\cv)},
\]
where $w_H(\cv)$ denotes the Hamming weight of the codeword $\cv.$  Then we have the following theorem.

\begin{theorem}
Let $C$ be a linear code of length $n$ over $R.$  Then
\[
Ham_C(X,Y)=SLWE_C(X,\underbrace{Y,Y,\ldots,Y}_{16}).
\]
\end{theorem}

\begin{proof}
Similar to the proof of Theorem 9 in \cite{Li17}.
\end{proof}

We also have the MacWilliams relation with respect to the Hamming weight enumerator.

\begin{theorem}
Let $C$ be a linear code of length $n$ over $R.$  Then
\[
Ham_{C^\perp}(X,Y)= \dfrac{1}{|C|}Ham_{C}(X+65536Y, X-Y).
\]
\end{theorem}

\begin{proof}
Similar to the proof of Theorem 10 in \cite{Li17}.
\end{proof}

Next, we consider the other weight enumerator with respect to the Lee weight, called Lee weight enumerator.
For a linear code $C,$ define $A_i$ as a number of elements of $C$ having Lee weight $i$.
The sequence $A_0,A_1,\ldots, A_{16n}$ is called \emph{weight distribution} in $C$ with respect to the Lee weight.
The \emph{Lee weight enumerator} for $C$ is defined by
\[
Lee_C(X,Y)=\sum\limits_{\textbf{c}\in C}X^{16n-w_L(\textbf{c})}Y^{w_L(\textbf{c})}=\sum\limits_{i=0}^{16n}A_iX^{16n-i}Y^{i}
\]
Then we have the following property.
\begin{theorem}\label{lee}
Let $C$ be a linear code of length $n$ over $R$. Then
\[
\begin{aligned}
		Lee_C(X,Y)&=SLWE_C(X^{16},X^{15}Y,X^{14}Y^{2},X^{13}Y^{3},X^{12}Y^{4},X^{11}Y^{5},X^{10}Y^{6},X^{9}Y^{7},X^{8}Y^{8},\\
		&~~~~~X^{7}Y^{9},X^{6}Y^{10},X^{5}Y^{11},X^{4}Y^{12},X^{3}Y^{13},X^{2}Y^{14},XY^{15},Y^{16}).
\end{aligned}
\]
\end{theorem}

\begin{proof}
Denote $w_L(\cv)=\sum\limits_{i=0}^{16}in_i(\cv)$.  Then we have
\[
16n-w_L(\cv)=\sum\limits_{i=0}^{16}16 n_i(\cv)-\sum\limits_{i=0}^{16}i n_i(\cv)
=\sum\limits_{i=0}^{16}(16-i) n_i(\cv).
\]
By definition, we obtain
\[
\begin{aligned}
		Lee_C (X,Y)&=\sum\limits_{\cv \in C}X^{16n-w_L(\cv)}Y^{w_L(\cv)}\\
		&=\sum\limits_{\cv\in C}X^{\sum\limits_{i=0}^{16}(16-i)n_i}Y^{\sum\limits_{i=0}^{16}i n_i}\\
		&=\sum\limits_{\cv \in C}\prod\limits_{i=0}^{16}X^{16-i}Y^i\\
		&=SLWE_C(X^{16},X^{15}Y,\ldots,Y^{16}).
\end{aligned}
\]
\end{proof}
The following theorem give us a MacWilliams relation with respect to the Lee weight enumerator.

\begin{theorem}
Let $C$ be a linear code of length  $n$ over $R$. Then
\[
Lee_{C^\perp}(X,Y)=\dfrac{1}{|C|}Lee_{C}(X+Y,X-Y).
\]
\end{theorem}

\begin{proof}
By Theorem \ref{slwe} dan Theorem \ref{lee}, we obtain
\[
\begin{aligned}
		Lee_{C^\perp}(X,Y)=SLWE_{C^\perp}(X^{16},X^{15}Y,\ldots,Y^{16})=\dfrac{1}{|C|}SLWE_C(E_0,E_1,E_2,\ldots,E_{16})
\end{aligned}
\]
with
\[
\begin{aligned}
		E_0 =~& X^{16}+16X^{15}Y+120X^{14}Y^{2}+560X^{13}Y^{3}+1280X^{12}Y^{4}+4368X^{11}Y^{5}+8008X^{10}Y^{6}\\
		&+11440X^{9}Y^{7}+12870X^{8}Y^{8}+11440X^{7}Y^{9}+8008X^{6}Y^{10}+4368X^{5}Y^{11}+1280X^{4}Y^{12}\\
		&+560X^{3}Y^{13}+120X^{2}Y^{14}+16XY^{15}+Y^{16}=(X+Y)^{16},\\
\end{aligned}
\]
\[
\begin{aligned}
		E_1 = ~&X^{16}+14X^{15}Y+90X^{14}Y^{2}+350X^{13}Y^{3}+910X^{12}Y^{4}+1638X^{11}Y^{5}+2002X^{10}Y^{6}\\
		&+1430X^{9}Y^{7}-1430X^{7}Y^{9}-2002X^{6}Y^{10}-1638X^{5}Y^{11}-910X^{4}Y^{12}-350X^{3}Y^{13}\\
		&-90X^{2}Y^{14}-14XY^{15}-Y^{16}=(X+Y)^{15}(X-Y),\\
\end{aligned}
\]
\[
\begin{aligned}
		E_2 = ~&X^{16}+12X^{15}Y+64X^{14}Y^{2}+196X^{13}Y^{3}+364X^{12}Y^{4}+364X^{11}Y^{5}-572X^{9}Y^{7}\\
		&-858X^{8}Y^{8}-572X^{7}Y^{9}+364X^{5}Y^{11}+364X^{4}Y^{12}+196X^{3}Y^{13}+64X^{2}Y^{14}\\
		&+12XY^{15}+Y^{16}=(X+Y)^{14}(X-Y)^{2},\\
\end{aligned}
\]
\[
\begin{aligned}
		E_3 = ~&X^{16}+10X^{15}Y+42X^{14}Y^{2}+90X^{13}Y^{3}+78X^{12}Y^{4}-78X^{11}Y^{5}-286X^{10}Y^{6}-286X^{9}Y^{7}\\
		&+286X^{7}Y^{9}+286X^{6}Y^{10}+78X^{5}Y^{11}-78X^{4}Y^{12}-90X^{3}Y^{13}-42X^{2}Y^{14}-10XY^{15}\\
		&-Y^{16}=(X+Y)^{13}(X-Y)^{3},\\
\end{aligned}
\]
\[
\begin{aligned}
		E_4 = ~&X^{16}+8X^{15}Y+24X^{14}Y^{2}+24X^{13}Y^{3}-36X^{12}Y^{4}-120X^{11}Y^{5}-88X^{10}Y^{6}+88X^{9}Y^{7}\\
		&+198X^{8}Y^{8}+88X^{7}Y^{9}-88X^{6}Y^{10}-120X^{5}Y^{11}-36X^{4}Y^{12}+24X^{3}Y^{13}+24X^{2}Y^{14}\\
		&+8XY^{15}+Y^{16}=(X+Y)^{12}(X-Y)^{4},\\
\end{aligned}
\]
\[
\begin{aligned}
		E_5 = ~&X^{16}+6X^{15}Y+10X^{14}Y^{2}-10X^{13}Y^{3}-50X^{12}Y^{4}-34X^{11}Y^{5}+66X^{10}Y^{6}+110X^{9}Y^{7}\\
		&-110X^{7}Y^{9}-66X^{6}Y^{10}+34X^{5}Y^{11}+50X^{4}Y^{12}+10X^{3}Y^{13}-10X^{2}Y^{14}-6XY^{15}\\
		&-Y^{16}=(X+Y)^{11}(X-Y)^{5},\\
\end{aligned}
\]
\[
\begin{aligned}
		E_6 = ~&X^{16}+4X^{15}Y-20X^{13}Y^{3}-20X^{12}Y^{4}+36X^{11}Y^{5}+64X^{10}Y^{6}-20X^{9}Y^{7}-90X^{8}Y^{8}\\
		&-20X^{7}Y^{9}+64X^{6}Y^{10}+36X^{5}Y^{11}-20X^{4}Y^{12}-20X^{3}Y^{13}+4XY^{15}\\
		&+Y^{16}=(X+Y)^{10}(X-Y)^6,\\
\end{aligned}
\]
\[
\begin{aligned}
		E_7 = ~&X^{16}+2X^{15}Y-6X^{14}Y^{2}-14X^{13}Y^{3}+14X^{12}Y^{4}+42X^{11}Y^{5}-14X^{10}Y^{6}-70X^{9}Y^{7}\\
		&+70X^{7}Y^{9}+14X^{6}Y^{10}-42X^{5}Y^{11}-14X^{4}Y^{12}+14X^{3}Y^{13}+6X^{2}Y^{14}\\
		&-2XY^{15}-Y^{16}=(X+Y)^9(X-Y)^7,\\
\end{aligned}
\]
\[
\begin{aligned}
		E_8 = ~&X^{16}-8X^{14}Y^{2}+28X^{12}Y^{4}-56X^{10}Y^{6}+70X^{8}Y^{8}-56X^{6}Y^{10}+28X^{4}Y^{12}-8X^{2}Y^{14}\\
		&+Y^{16}=(X+Y)^8(X-Y)^8,\\
\end{aligned}
\]
\[
\begin{aligned}
		E_9 = ~&X^{16}-2X^{15}Y-6X^{14}Y^{2}+14X^{13}Y^{3}+14X^{12}Y^{4}-42X^{11}Y^{5}-14X^{10}Y^{6}+70X^{9}Y^{7}\\
		&-70X^{7}Y^{9}+14X^{6}Y^{10}+42X^{5}Y^{11}-14X^{4}Y^{12}-14X^{3}Y^{13}+6X^{2}Y^{14}+2XY^{15}\\
		&-Y^{16}=(X+Y)^7(X-Y)^9,\\
\end{aligned}
\]
\[
\begin{aligned}
		E_{10} = ~&X^{16}-4X^{15}Y+20X^{13}Y^{3}-20X^{12}Y^{4}-36X^{11}Y^{5}+64X^{10}Y^{6}+20X^{9}Y^{7}\\
		&-90X^{8}Y^{8}+20X^{7}Y^{9}+64X^{6}Y^{10}-36X^{5}Y^{11}-20X^{4}Y^{12}+20X^{3}Y^{13}-4XY^{15}\\
		&+Y^{16}=(X+Y)^6(X-Y)^{10},\\
\end{aligned}
\]
\[
\begin{aligned}
		E_{11} = ~&X^{16}-6X^{15}Y+10X^{14}Y^{2}+10X^{13}Y^{3}-50X^{12}Y^{4}+34X^{11}Y^{5}+66X^{10}Y^{6}-110X^{9}Y^{7}\\
		&+110X^{7}Y^{9}-66X^{6}Y^{10}-34X^{5}Y^{11}+50X^{4}Y^{12}-10X^{3}Y^{13}-10X^{2}Y^{14}+6XY^{15}\\
		&-Y^{16}=(X+Y)^5(X-Y)^{11},\\
\end{aligned}
\]
\[
\begin{aligned}
		E_{12} = ~&X^{16}-8X^{15}Y+24X^{14}Y^{2}-24X^{13}Y^{3}-36X^{12}Y^{4}+120X^{11}Y^{5}-88X^{10}Y^{6}-88X^{9}Y^{7}\\
		&+198X^{8}Y^{8}-88X^{7}Y^{9}-88X^{6}Y^{10}+120X^{5}Y^{11}-36X^{4}Y^{12}-24X^{3}Y^{13}+24X^{2}Y^{14}\\
		&-8XY^{15}+Y^{16}=(X+Y)^4(X-Y)^{12},\\
\end{aligned}
\]
\[
\begin{aligned}
		E_{13} = ~&X^{16}-10X^{15}Y+42X^{14}Y^{2}-90X^{13}Y^{3}+78X^{12}Y^{4}+78X^{11}Y^{5}-286X^{10}Y^{6}\\
		&+286X^{9}Y^{7}-286X^{7}Y^{9}+286X^{6}Y^{10}-78X^{5}Y^{11}-78X^{4}Y^{12}+90X^{3}Y^{13}\\
		&-42X^{2}Y^{14}+10XY^{15}-Y^{16}=(X+Y)^3(X-Y)^{13},\\
\end{aligned}
\]
\[
\begin{aligned}
		E_{14} = ~&X^{16}-12X^{15}Y+64X^{14}Y^{2}-196X^{13}Y^{3}+364X^{12}Y^{4}-364X^{11}Y^{5}+572X^{9}Y^{7}\\
		&-858X^{8}Y^{8}+572X^{7}Y^{9}-364X^{5}Y^{11}+364X^{4}Y^{12}-196X^{3}Y^{13}+64X^{2}Y^{14}\\
		&-12XY^{15}+Y^{16}=(X+Y)^2(X-Y)^{14},\\
\end{aligned}
\]
\[
\begin{aligned}
		E_{15} = ~&X^{16}-14X^{15}Y+90X^{14}Y^{2}-350X^{13}Y^{3}+910X^{12}Y^{4}-1638X^{11}Y^{5}+2002X^{10}Y^{6}\\
		&-1430X^{9}Y^{7}+1430X^{7}Y^{9}-2002X^{6}Y^{10}+1638X^{5}Y^{11}-910X^{4}Y^{12}+350X^{3}Y^{13}\\
		&-90X^{2}Y^{14}+14XY^{15}-Y^{16}=(X+Y)(X-Y)^{15},\\
\end{aligned}
\]
\[
\begin{aligned}
		E_{16} = ~&X^{16}-16X^{15}Y+120X^{14}Y^{2}-560X^{13}Y^{3}+1280X^{12}Y^{4}-4368X^{11}Y^{5}+8008X^{10}Y^{6}\\
		&-11440X^{9}Y^{7}+12870X^{8}Y^{8}-11440X^{7}Y^{9}+8008X^{6}Y^{10}-4368X^{5}Y^{11}\\
		&+1280X^{4}Y^{12}-560X^{3}Y^{13}+120X^{2}Y^{14}-16XY^{15}+Y^{16}=(X-Y)^{16}.
\end{aligned}
\]
Hence, we have
\[
Lee_{C^\perp}(X,Y)=\dfrac{1}{|C|}Lee(X+Y,X-Y).
\]
\end{proof}

\section{Cyclic and quasi-cyclic Codes}
Now, let us look at an important class of linear codes, namely cyclic codes. In this section we mainly consider
the structural properties of cyclic codes over the ring $R.$

The notion of cyclic codes is standard for codes over all rings. A cyclic shift on $R^n$ is a permutation $T$
such that
\[
T(c_0,c_1,c_2,\ldots, c_{n-1})=(c_{n-1},c_0,c_1,\ldots,c_{n-2}).
\]
A linear code $C$ over $R$ is called \emph{cyclic code} if $C$ is invariant under the cyclic shift $T,$ namely
$T(C)=C.$  We use the usual ideas of identifying vectors in $R^n$ and polynomials in
the residue class ring $R[x]/ \langle x^n-1 \rangle$ as follows:
\[
\cv=(c_0,c_1,c_2,\ldots,c_{n-1})  \longleftrightarrow
c(x)=c_0+c_1x+c_2x^2+\cdots+c_{n-1}x^{n-1} + \langle x^n-1 \rangle.
\]
With this identification we see that $T(\cv)$ is identified by $x \cdot c(x) \in R[x]/ \langle x^n-1 \rangle.$
It implies that cyclic codes over $R$ are identified by ideals in the residue class ring $R[x]/ \langle x^n-1 \rangle.$
Hence, in order to understand cyclic codes over the ring $R$ it is essential for us to understand the structure
of the residue class ring $R[x]/ \langle x^n-1 \rangle.$

The first theorem below is a straightforward generalization of Theorem 13 proven by Li, Guo, Zhu, and Kai \cite{Li17}.

\begin{theorem}
Let $\DS C=C_1\eta_1 \oplus C_2 \eta_2 \oplus \cdots \oplus C_8 \eta_8.$
Then $C$ is a cyclic code over $R$
if and only if one of following three conditions is satisfied:
\begin{enumerate}[(1)]
\item For $t\in \{1,2,\ldots,8\},$ $C_t$ is a cyclic code over $\mathbb{Z}_4.$
\item For $t\in \{1,2,\ldots,8\},$ $C_t^{\perp}$ is a cyclic code over $\mathbb{Z}_4.$
\item $C^{\perp}$ is a cyclic code over $R$.
\end{enumerate}
\end{theorem}

\begin{proof}
Let $\DS \cv=\sum_{t=1}^8 \eta_t\cv_t\in C$, and write $\cv_t=(c_{t,0},c_{t,1},\ldots,c_{t,n-1})\in C_t,$
for $1 \leq t \leq 8.$
Since $C$ is a cyclic code, we also have
\[
\left(\sum_{t=1}^8 \eta_t c_{t,n-1},\sum_{t=1}^8 \eta_t c_{t,0},\ldots,\sum_{t=1}^8 \eta_t c_{t,n-2}\right)\in C.
\]
So, $(c_{t,n-1},c_{t,0},\dots,c_{t,n-2})\in C_t,$ for $1 \leq t \leq 8,$ and hence, $C_t$ is cyclic for $1 \leq t \leq 8.$
The reverse also holds, so the first condition is proven.

If $C_t$ is cyclic over $\ZZ_4,$ then $C_t^{\perp}$ is also cyclic (\cite{Wan97}, Proposition 7.9).
From $(1),$ we have $C^{\perp}$ is a cyclic code over $R,$ so we have $C$ is a cyclic code over $R.$
\end{proof}

We start to observe the generator polynomials of cyclic code and its dual over $R.$  For that purpose, we need
the following theorem proven by Li, Guo, Zhu, and Kai \cite{Li17}.

\begin{theorem}[\cite{Li17}, Theorem 15]\label{T-Li}
Let $C=\langle f(x)+2p(x),~2g(x) \rangle$ be a cyclic code over $\ZZ_4.$  Then
\[
C^{\perp}=\langle \widehat{g}(x)^*+2x^{\deg(\widehat{g}(x))-\deg(u(x))}u(x)^*,~2\widehat{f}(x)^*\rangle
\]
with $\DS\widehat{f}(x):=\left(\frac{x^n-1}{f(x)} \right),$ $\DS\widehat{g}(x):=\left(\frac{x^n-1}{g(x)} \right),$
 and $\DS f(x)^*:=x^{\deg(f(x))}f\left(\frac{1}{x}\right)$.
\end{theorem}

The following two theorems provide generator polynomials of cyclic code and its dual over $R.$

\begin{theorem}\label{T-CC}
Let $\DS C=C_1 \eta_1 \oplus C_2 \eta_2 \oplus \cdots \oplus C_8 \eta_8$
be a cyclic code of length $n$ over $R$. If for every $t\in\{1,2,\dots,8\}$,
there exist polynomials $f_t(x),g_t(x),p_t(x)\in\mathbb{Z}_4[x]$ such that $C_t=\langle f_t(x)+2p_t(x),~2g_t(x) \rangle,$ then
\[
\DS C=\left \langle \sum_{t=1}^8\eta_tf_t(x)+2\sum_{t=1}^8\eta_tp_t(x),~2\sum_{t=1}^8\eta_tg_t(x) \right\rangle.
\]
Furthermore, if $n$ is odd, then $\DS C=\left \langle \sum_{t=1}^8\eta_tf_t(x)+2\sum_{t=1}^8\eta_tg_t(x)\right \rangle.$
\end{theorem}

\begin{proof}
Let $\DS D=\left \langle \sum_{t=1}^8\eta_tf_t(x)+2\sum_{t=1}^8\eta_tp_t(x),~2\sum_{t=1}^8\eta_tg_t(x)\right \rangle.$
It is obvious that $D\subseteq C$. Let $c(x)\in C$. Because $\DS C=\oplus_{t=1}^8 \eta_tC_t$
and $C_t=\langle f_t(x)+2p_t(x),~2g_t(x) \rangle,$ then there exist $u_t(x),v_t(x)\in \mathbb{Z}_2[x]$ such that
\begin{align*}
c(x)&=\sum_{t=1}^8\eta_t((f_t(x)+2p_t(x))u_t(x)+2g_t(x)v_t(x))\\
    &=\sum_{t=1}^8\eta_t(f_t(x)+2p_t(x))u_t(x)+\sum_{t=1}^8\eta_t2g_t(x)v_t(x)\\
    &=\sum_{t=1}^8\eta_tu_t(x)\sum_{t=1}^8\eta_t(f_t(x)+2p_t(x))+\sum_{t=1}^8\eta_tv_t(x)\sum_{t=1}^82\eta_tg_t(x).
\end{align*}
So we have $C\subseteq D$ and hence $C=D$.
\end{proof}

By using Theorem \ref{T-Li} and the similar technique as in proof of Theorem \ref{T-CC}, we obtain generator polynomials
for the dual of cyclic codes as given in the theorem below.

\begin{theorem}
Let $C=\langle f(x)+2p(x),~2g(x) \rangle$ be a cyclic code over $\ZZ_4.$ Then
\[
C^{\perp}=\left \langle \sum_{t=1}^8\eta_t\widehat{g}_t(x)^*+2\sum_{t=1}^8\eta_tx^{\deg(\widehat{g}_t(x))-\deg(u_t(x))}u_t(x)^*,
~2\sum_{t=1}^8\widehat{f}_t(x)^* \right \rangle.
\]
\end{theorem}

Now, let us turn to the special class of cyclic codes called quasi-cyclic codes.

Let $\sigma$ be a cyclic shift operator over $\ZZ_4^n.$
For any positive integer $s,$ let $\sigma_s$ be the quasi-shift defined by
\[
\sigma_s\left(a^{(1)} \mid a^{(2)} \mid \cdots \mid a^{(s)}\right)
=\left(\sigma\left(a^{(1)}\right) \mid \sigma\left(a^{(2)}\right) \mid \cdots \mid \sigma\left(a^{(s)}\right)\right),
\]
with $a^{(1)},a^{(2)},\ldots,a^{(s)}\in\ZZ_4^n$ and $"|"$ is a vector concatenation.
A quaternary \emph{quasi-cyclic code} $C$ of index $s$ and length $ns$ is a subset of $(\ZZ_4^n)^s$
such that $\DS \sigma_s(C)=C$. If $\DS R=\oplus_{t=1}^8 \eta_tR_t,$ we can write any $r\in R$ as
$\DS r=\sum_{t=1}^8\eta_tr_t$ with $r_t\in R_t,$ for $1 \leq t \leq 8.$
We define the mapping
\begin{center}
$\Phi: R^n \longrightarrow \left(\mathbb{Z}_4^{2^3}\right)^n$\\
$~~\DS \times_{i=0}^n r_i\longmapsto \times_{t=1}^8\times_{i=0}^{n-1} r_{t,i}$
\end{center}
with
$\DS r_i=\times_{t=1}^8 r_{t,i}$ for $i=0,1,\ldots,n-1$ and $r_{t,i}\in R_t.$

Then we have a similar theorem of Theorem 17 in \cite{Li17}.

\begin{theorem}
Let $\DS C=C_1 \eta_1 \oplus C_2 \eta_2 \oplus \cdots \oplus C_8 \eta_8$
be a cyclic code of length $n$ over $R.$
Then $\Phi(C)$ is a quasi-cyclic code of index $8$ and length $8n$ over $\ZZ_4.$
\end{theorem}

\begin{proof}
Let $\times_{i=0}^{n-1} c_i \in C.$ Let $c_i=\times_{t=1}^8 c_{t,i}$
for $i=0,1,\dots,n-1$ and $c_{t,i}\in C_t$. Since $C$ is cyclic code,
we have $C_t$ is cyclic for $ 1 \leq t \leq 8.$
This means that for every $t\in\{1,2,\dots,8\},$ we have $\sigma(\times_{i=0}^{n-1} r_{t,i})\in C_t,$
if $\times_{i=0}^{n-1} r_{t,i}\in C_t.$
Write $\Phi(\times_{i=0}^{n-1} c_i)=\times_{t=1}^8\times_{i=0}^{n-1} r_{t,i}.$
Then
\[
\sigma_{8}(\times_{t=1}^8\times_{i=0}^{n-1} r_{t,i})=\times_{t=1}^8\sigma(\times_{i=0}^{n-1} r_{t,i})\in\Phi(C).
\]
So we have $\Phi(C)$ is a quasi-cyclic code $C$ of index $8$ and length $8n$ over $\ZZ_4.$
\end{proof}

Furthermore, by using the Theorem 18 of \cite{Li17} below, we obtain directly the type of $\Phi(C)$ as given
in Corollary \ref{C-CC}.

\begin{theorem}[\cite{Li17}, Theorem 18]
Let $C_t,t\in\{1,2,\dots,8\}$ be a cyclic code of length $n$ ($n$ is odd) over $\ZZ_4.$
Write $C_t=\langle f_{1,t}(x)+2f_{2,t}(x) \rangle$ with $f_{1,t}(x)$ and $f_{2,t}(x)$ are monic
factors of $x^n-1$ over $\ZZ_4$ and $f_{2,t}(x) \mid f_{1,t}(x).$
Then the cardinality of $C_t,$ for $1 \leq t \leq 8,$ is
\[
4^{n-\deg(f_{1,t}(x))}2^{\deg(f_{1,t}(x)-\deg(f_{2,t}(x))}.
\]
\end{theorem}

The corollary below follows directly.
\begin{cor}\label{C-CC}
Let $\DS\Phi(C)=\prod_{t=1}^8 C_t$ be a linear code of length $8n$ ($n$ is odd)
over $\ZZ_4$ and $C_t$ is a cyclic code over $\ZZ_4$ for every $t\in\{1,2,\ldots,8\}.$
Then the cardinality of $\Phi(C)$ is
\[
\DS 4^{\sum_{t=1}^8 (n-\deg(f_{1,t}(x)))}2^{\sum_{t=1}^8 (\deg(f_{1,t}(x)-\deg(f_{2,t}(x))}.
\]
\end{cor}

\subsection{Some examples}

Here we provide some examples of cyclic codes of odd length over $R$ and their $\ZZ_4$-images.

\begin{ex}
Let $n=3$. In $\ZZ_4[x],$ $x^3-1=(x - 1) (x^2 + x + 1)$.

Choose $C_i=\langle (x^2 + x + 1)+2 \rangle$ for $i=1,2,\ldots,8.$ We have $\DS C=\oplus_{t=1}^8C_t \eta_t.$
Parameters of $\Phi(C)$ is $[24,4^82^{16},2].$ $\blacktriangleleft$
\end{ex}

\begin{ex}
Let $n=5$. In $\ZZ_4[x],$ $x^5-1=(x - 1) (x^4+x^3+x^2 + x + 1).$
Choose $C_i=\langle (x^4+x^3+x^2 + x + 1)+2 \rangle$ for $i=1,2,\ldots,8$. We have $\DS C=\oplus_{t=1}^8C_t \eta_t.$
Parameters of  $\Phi(C)$ is $[40,4^82^{32},2].$
$\blacktriangleleft$
\end{ex}

\begin{ex}
Let $n=7$. In $\ZZ_4[x],$ $x^7-1=(x + 1) (x^3 + x + 1) (x^3 + x^2 + 1)$.
Choose
\begin{align*}
    C_1&=C_2=C_3=\langle (x^3 + x + 1)+2 \rangle,\\
    C_4&=C_5=C_6=C_7=C_8=\langle (x^3 + x^2 + 1)+2 \rangle.
\end{align*}
Then we have
$\DS C=\oplus_{t=1}^8C_t \eta_t$
is a cyclic code over $R.$ Parameters of $\Phi(C)$ is $[56,4^{32}2^{24},2].$

If we choose another set of $C_i$ with
\[
C_i=\langle (x^3 + x + 1)(x^3 + x^2 + 1)+2 \rangle, ~i=1,2,\ldots,8,
\]
then we have
$\DS C=\oplus_{t=1}^8C_t\eta_t$
is also a cyclic code over $R.$ Parameters of $\Phi(C)$ is $[56,4^{8}2^{48},6].$
Let us choose
\begin{align*}
    C_1&=C_2=C_3=\langle (x + 1) (x^3 + x + 1) (x^3 + x^2 + 1)+2(x^3 + x + 1) \rangle,\\
    C_4&=C_5=C_6=C_7=C_8=\langle (x^3 + x^2 + 1)+2 \rangle.
\end{align*}
We have
$\DS C=\oplus_{t=1}^8C_t\eta_t$
is also a cyclic code over $R.$ Parameters of $\Phi(C)$ is $[56,4^{23}2^{27},2].$
$\blacktriangleleft$
\end{ex}

\begin{ex}
Let $n=9.$ In $\ZZ_4[x]$, $x^9-1=(x + 1) (x^2 + x + 1)(x^6+x^3+1).$
Choose
\begin{center}
    $C_i=\langle (x^2 + x + 1)(x^6+x^3+1)+2(x^6+x^3+1) \rangle, ~i=1,2,\ldots,8, $
\end{center}
then $\DS C=\oplus_{t=1}^8C_t\eta_t$
is a cyclic code over $R$. Parameters of $\Phi(C)$ is $[72,4^{8}2^{16},6].$

If we choose
\[
C_i=\langle (x^2 + x + 1)(x^6+x^3+1)+2(x^2 + x + 1) \rangle, ~i=1,2,\ldots,8,
\]
We have
$\DS C=\oplus_{t=1}^8C_t\eta_t$
is a cyclic code over $R.$ Parameters of $\Phi(C)$ is $[72,4^{8}2^{32},3].$
$\blacktriangleleft$
\end{ex}

\begin{ex}
Let $n=15.$ In $\ZZ_4[x],$ $x^{15}-1=(x + 1) (x^2 + x + 1)(x^4 + x + 1)(x^4+x^3+1)(x^4+x^3+x^2+1).$

Choose
\begin{align*}
    C_i&=\langle (x^4 + x + 1)(x^4+x^3+1)(x^4+x^3+x^2+1)\\
       &~~+2(x^4 + x + 1)(x^4+x^3+1) \rangle,~\text{for }1 \leq i \leq 4,\\
    C_i&=\langle (x^4 + x + 1)(x^4+x^3+1)(x^4+x^3+x^2+1)\\
       &~~+2(x^4+x^3+1)(x^4+x^3+x^2+1) \rangle,~\text{for }5 \leq i \leq 8.
\end{align*}
We have $\DS C=\oplus_{t=1}^8C_t\eta_t$
is a cyclic code over $R.$ Parameters of $\Phi(C)$ is $[120,4^{24}2^{32},10].$

If we choose
\begin{align*}
    C_i&=\langle (x^2 + x + 1)(x^4 + x + 1)(x^4+x^3+1)\\
       &~~+2(x^4 + x + 1)(x^4+x^3+1) \rangle,~\text{for }1 \leq i \leq 4,\\
    C_i&=\langle (x^2 + x + 1)(x^4+x^3+1)(x^4+x^3+x^2+1)\\
       &~~+2(x^4+x^3+1)(x^4+x^3+x^2+1) \rangle~\text{for }5 \leq i \leq 8.
\end{align*}
then $\DS C=\oplus_{t=1}^8C_t \eta_t$
is a cyclic code over $R.$ Parameters of $\Phi(C)$ is $[120,4^{40}2^{16},8].$
$\blacktriangleleft$
\end{ex}

\begin{ex}
Let $n=31.$ In $\ZZ_4[x],$ $x^{31}-1=F_1(x) F_2(x) F_3(x)F_4(x)F_5(x)F_6(x)F_7(x)$ with
\begin{align*}
    F_1(x)&=x+1\\
    F_2(x)&=x^5 + x^2 + 1\\
    F_3(x)&=x^5 + x^3 + 1\\
    F_4(x)&=x^5 + x^3 + x^2 + x + 1\\
    F_5(x)&=x^5 + x^4 + x^2 + x + 1\\
    F_6(x)&=x^5 + x^4 + x^3 + x + 1\\
    F_7(x)&=x^5 + x^4 + x^3 + x^2 + 1.
\end{align*}
Choose
\begin{align*}
C_1&=C_2=\langle F_1(x)F_2(x)F_3(x)+2F_1(x)F_2(x) \rangle,\\
C_3&=\langle F_1(x)F_3(x)F_4(x)+2F_1(x)F_3(x) \rangle,\\
C_4&=\langle F_1(x)F_4(x)F_5(x)+2F_1(x)F_4(x) \rangle,\\
C_5&=\langle F_1(x)F_5(x)F_6(x)+2F_1(x)F_5(x) \rangle,\\
C_6&=\langle F_1(x)F_6(x)F_7(x)+2F_1(x)F_6(x) \rangle,\\
C_7&=C_8=\langle F_1(x)F_2(x)F_7(x)+2F_1(x)F_7(x)\rangle.
\end{align*}
We have  $\DS C=\oplus_{t=1}^8C_t \eta_t$ is a cyclic code over $R.$ Parameters of $\Phi(C)$ is $[248,4^{160}2^{40},8].$

For another set of $C_i$ such as
\begin{align*}
C_1&=C_2=C_3=\langle F_1(x)F_2(x)F_3(x)F_4(x)F_5(x)+2F_1(x)F_2(x)F_3(x) \rangle,\\
C_4&=C_5=C_6=\langle F_1(x)F_3(x)F_4(x)F_5(x)F_6(x)+2F_1(x)F_3(x)F_4(x) \rangle,\\
C_7&=C_8=\langle F_1(x)F_4(x)F_5(x)F_6(x)F_7(x)+2F_1(x)F_4(x)F_5(x) \rangle,
\end{align*}
we have  $\DS C=\oplus_{t=1}^8C_t\eta_t$ is a cyclic code over $R$. Parameters of $\Phi(C)$ is $[248,4^{80}2^{80},12].$

Let us choose another set of $C_i:$
\begin{align*}
C_1&=C_2=C_3=C_4=\langle F_2(x)F_3(x)F_4(x)F_5(x)F_6(x)+2F_3(x)F_4(x)F_5(x)F_6(x) \rangle,\\
C_5&=C_6=C_7=C_8=\langle F_3(x)F_4(x)F_5(x)F_6(x)F_7(x)+2F_4(x)F_5(x)F_6(x)F_7(x) \rangle.
\end{align*}
We have  $\DS C=\oplus_{t=1}^8C_t\eta_t$ is a cyclic code over $R.$ Parameters of $\Phi(C)$ is $[248,4^{48}2^{32},22].$
$\blacktriangleleft$
\end{ex}

\begin{remark}
We compare our results on linear codes over $\ZZ_4$ with the database of $\ZZ_4$ codes available online \cite{database}.
We conclude that the resulting linear codes are all new with the highest known minimum distances.
\end{remark}

\section{Conclusion}
In this paper we derive structural properties of linear codes over the ring
$R:=\ZZ_4+u\ZZ_4+v\ZZ_4+w\ZZ_4+uv\ZZ_4+uw\ZZ_4+vw\ZZ_4+uvw\ZZ_4.$  We also obtained some new and optimal
linear codes having parameters which are unknown to exist before.

There are several direction to further research on the codes over the ring.
We are now observing the self-duality as well as
polycyclic codes over the ring $R.$  We obtained structural properties regarding self-dual codes as well as constacyclic
codes over $R.$  The results, which are not included here, will be published elsewhere in separate
papers.

\section*{Acknowledgement}
This research is supported by \emph{Riset P3MI ITB 2018.}
A part of this work was done while the third author
visited Research Center for Pure and Applied Mathematics (RCPAM), Tohoku University, Japan on July 2018
under the financial support from \emph{Penelitian Berbasis Kompetensi Kemenristekdikti}.
The third author thanks Prof. Hajime Tanaka for warm hospitality.

\end{document}